%% file: JSTSP_revision2.tex
\documentclass[10pt,twocolumn,twoside]{IEEEtran}

\usepackage{cite}
\ifCLASSINFOpdf
   \usepackage[pdftex]{graphicx}
   \graphicspath{{../pdf/}{../jpeg/}}
   \DeclareGraphicsExtensions{.pdf,.jpeg,.png}
\else
   \usepackage[dvips]{graphicx}
   \graphicspath{{../eps/}}
   \DeclareGraphicsExtensions{.eps}
\fi

\usepackage{subfigure}
\usepackage{amsmath,amssymb,amsfonts,amstext,amsthm,mathrsfs,bm}
\usepackage{latexsym,graphics,epsf,epsfig,}
\usepackage{multirow}

\usepackage{algorithm,algorithmic}
\usepackage{graphicx,subfigure}
\usepackage{cite}
\usepackage{setspace}
\usepackage{wrapfig}
\usepackage{color}
\usepackage{caption}
\usepackage{epstopdf}
\usepackage{thmtools}
\usepackage{cuted}
\usepackage{lipsum}
\usepackage{dsfont}

\usepackage{cleveref}
\usepackage{xspace}

\usepackage{amsfonts, amsmath, amssymb, epsfig, graphicx, psfrag, subfigure, cite}
\usepackage{itrans}

\newcommand{\x}{\mathbf{x}}
\newcommand{\y}{\mathbf{y}}

\newcommand{\barD}{\bar{D} }
\newcommand{\K}{\mathcal{K}}

\newtheorem{definition}{\textbf{Definition}}
\newtheorem{corollary}{\textbf{Corollary}}
\newtheorem{theorem}{\textbf{Theorem}}
\newtheorem{proposition}{\textbf{Proposition}}
\newtheorem{remark}{\textbf{Remark}}

\newtheorem{lemma}{\textbf{Lemma}}

\newcommand{\nn}{\nonumber}

\newcommand{\mE}{\mathbb{E}}
\newcommand{\cH}{\mathcal{H}}
\newcommand{\cM}{\mathcal{M}}
\newcommand{\mmd}{\text{MMD}}

\begin{document}
\title{Nonparametric Composite Hypothesis Testing in an Asymptotic Regime}

\author{Qunwei~Li,~\IEEEmembership{Student~Member,~IEEE}, Tiexing~Wang, Donald~J.~Bucci,  Yingbin~Liang,~\IEEEmembership{Senior~Member,~IEEE},
Biao~Chen,~\IEEEmembership{Fellow,~IEEE} and Pramod~K.~Varshney,~\IEEEmembership{Life~Fellow,~IEEE}%
\thanks{This material is based upon work supported in part by the Defense Advanced
	Research Projects Agency under Contract No. HR0011-16-C-0135. The work of Y. Liang was also supported in part by the National Science Foundation under grant CCF-1801855. The work of B. Chen was also supported in part by the DDDAS program of the Air Force Office of Scientific Research under grant number FA9550-16-1-0077 and by the National Science Foundation under grant CNS-1731237.
	}
\thanks{Q.\ Li, T.\ Wang, B.\ Chen, and P.~K.\ Varshney are with the Department of Electrical Engineering and Computer Science, Syracuse University, Syracuse, NY 13244, USA
(e-mail: qli33@syr.edu; twang17@syr.edu; bichen@syr.edu; varshney@syr.edu).}
\thanks{D.~J.\ Bucci is with Lockheed Martin - Advanced Technology Labs, Cherry
Hill, NJ 08002, USA (e-mail: Donald.J.Bucci.Jr@lmco.com).}
\thanks{Y.\ Liang is with the Department of Electrical and Computer Engineering,
The Ohio State University, Columbus, OH 43210, USA (e-mail: liang.889@osu.edu).}
}

\maketitle
\begin{abstract}
We investigate the nonparametric, composite hypothesis testing problem for arbitrary unknown distributions in the asymptotic regime where both the sample size and the number of hypotheses grow exponentially large. Such asymptotic analysis is important in many practical problems, where the number of variations that can exist within a family of distributions can be countably infinite. We introduce the notion of \emph{discrimination capacity}, which captures the largest exponential growth rate of the number of hypotheses relative to the sample size so that there exists a test with asymptotically vanishing probability of error. Our approach is based on various distributional distance metrics in order to incorporate the generative model of the data. We provide analyses of the error exponent using the maximum mean discrepancy (MMD) and Kolmogorov-Smirnov (KS) distance and characterize the corresponding discrimination rates, i.e., lower bounds on the discrimination capacity, for these tests. Finally, an upper bound on the discrimination capacity based on Fano's inequality is developed. Numerical results are presented to validate the theoretical results.

\end{abstract}

\input{Intro}
\input{problem_formulation}

\input{analysis_results}

\input{simulation}

\input{conclusion}

\input{appendix}


\bibliographystyle{IEEEtran}
\bibliography{darpa}

\end{document}

%% file: Intro.tex

\section{Introduction}

Information theory was largely developed in the context of communication systems, where information theoretic tools play an important role in characterizing the performance limits of such systems. However, another important area where information theoretic approach has proved useful is in statistical inference, e.g., hypothesis testing. For {\em parametric} hypothesis testing problems, information theoretic tools such as joint typicality, the equipartition property, and Sanov's theorem have been developed to characterize the error exponent \cite{Cover06,Gallager68,Csiszar81}. Information theory has also been applied to investigate a class of {\em parametric} hypothesis testing problems \cite{Han87,Han98}, where correlated data samples are observed over multiple terminals and data compression needs to be carried out in a decentralized manner. Additionally, information theory has also been applied to solve {\em nonparametric} hypothesis testing problems under the Neyman-Pearson framework \cite{Gutman89,Levitan02}.

In this paper, we apply information theoretic tools to study the nonparametric hypothesis testing problem, but with a focus on the average error probability instead of the Neyman-Pearson formulation. We address a more general scenario, where each hypothesis corresponds to a cluster of distributions. Such a nonparametric problem has not been thoroughly explored in the literature. We develop two nonparametric tests based respectively on the maximum mean discrepancy (MMD) and the Kolmogorov-Smirnov (KS) distance, and characterize the {\em exponential} error decay rate for these tests. Furthermore, in contrast to previous works where the number of hypotheses is assumed to be fixed, we study the regime where the number of hypotheses scales along with the sample size. This is analogous to the information theoretic channel coding problem where the number of messages scales along with the codeword length. Hence, in our study, information theory not only provides a technical tool to analyze the performance, but also provides an asymptotic perspective for understanding nonparametric hypothesis testing problems in the regime where the number of hypotheses is large, i.e., in the large-hypothesis large-sample regime.

More specifically, this paper assumes that there are $M$ hypotheses, each corresponding to a (cluster of) distributions, which are {\em unknown}. Sequences of length-$n$ training data samples generated by each distribution are available. {The more general case with the training data sequences having different lengths is discussed in Section \ref{sec:varyinglength}}. Suppose a length-$n$ test data stream is observed, which are samples generated by one of the distributions. The goal is to determine the cluster that contains the distribution that generated the observed test sequence. We are interested in the large-hypothesis regime, in which $M=2^{nD}$, i.e., the number of hypotheses scales exponentially in the {number of samples with a constant rate $D$}. The analogy to the channel coding problem \cite{Cover06} is now apparent where the exponent represents the transmission rate, i.e., the transmitted bits per channel use, for the channel coding problem, here $D$ represents the number of {\em hypothesis bits} that can be distinguished per observation sample. Correspondingly, we refer to $D$ as the {\em discrimination rate}, and the largest such value is referred to as the {\em discrimination capacity}. The notion of {\em discrimination capacity} provides the fundamental performance limit for the hypothesis testing problem in the large-hypothesis and large-sample regime.

\subsection{Main Contributions}
This paper makes the following major contributions.
\begin{itemize}

\item We provide an asymptotic viewpoint to understand the nonparametric hypothesis testing problem in the regime where the number of hypotheses scales exponentially in the sample size. Based on its connection to the channel coding problem, we introduce the notions of the discrimination rate and the discrimination capacity as the performance metrics in such an asymptotic regime.	
\item We develop two nonparametric approaches to solve the hypothesis testing problem that are based respectively on the maximum mean discrepancy (MMD) and the Kolmogorov-Smirnov (KS) distance. For both tests, we derive the corresponding error exponents and the discrimination rates. Our results show that as long as the number $M$ of hypotheses does not scale too fast, i.e., the scaling (discrimination) exponent is less than a certain threshold, the derived tests are exponentially consistent, i.e., the error probability converges to zero exponentially fast.
\item We also derive an upper bound on the discrimination capacity, which serves as an upper limit beyond which exponential consistency cannot be achieved by any nonparametric composite hypothesis testing rule.
\end{itemize}

\subsection{Related Work}

\textbf{$M$-ary hypothesis testing:} For {\em parametric} hypothesis testing problems, information theoretic tools have been developed to characterize the error exponent \cite{Cover06,Gallager68,Csiszar81,westover2008achievable}, and to study a class of distributed parametric hypothesis testing problems \cite{Han87,Han98,vazquez2016bayesian}. For sequential multi-hypothesis testing, information theoretic bounds on the sample size subject to constraints on the error probabilities have been developed in \cite{825826}. A generalization of the classical hypothesis testing problem is studied in \cite{ 6512566}, where a Bayesian decision maker is designed to enhance its information about the correct hypothesis. 
Information theory has also been applied to study {\em nonparametric} hypothesis testing problems with the primary focus being on the Neyman-Pearson formulation \cite{Gutman89,Levitan02}. An information-theoretic approach to the problem of a nonparametric hypothesis test with a Bayesian formulation is presented in \cite{ 1315943}. By factorizing dependent variables into mutually independent subsets, it has been shown that the likelihood ratio can be written as the sum of two sets of Kullback-Leibler divergence (KLD) terms, which is then used to quantify loss in hypothesis separability. Our study is different from the previous studies on nonparametric hypothesis testing problems in that we focus on the asymptotic regime where the number of hypotheses scales with sample size.

\textbf{Supervised learning:} The problem we study here can also be viewed as a supervised learning problem studied in the machine learning literature. However, the problem formulated here is different from the traditional supervised learning problem \cite{Bishop06}, where sample points corresponding to the same label are simply treated as individual samples, and their underlying statistical structure is not exploited in the design of classification rules. For example, the support vector machine (SVM) is one of the important classification algorithms for supervised learning, where the distance between samples is measured either by the Euclidean distance or by a kernel-based distance. Such distances do not exploit the underlying statistical distributions of data samples. {A robust form of the SVM in \cite{shivaswamy2006second} incorporates the probabilistic
uncertainty into the maximization of the margin. Our formulation exploits the underlying probabilistic structure of data samples, which is also robust to missing data, system noise, etc.}

A formulation of the supervised learning problem that is similar to our formulation has been studied previously in \cite{Muan2012}. The proposed approach, therein named support measure machine (SMM), exploits the kernel mean embedding to estimate the distance between probability distributions. In fact, the comparison between an SMM and an SVM also reflects the differences between our formulation and the traditional supervised learning problem.
However, the study in \cite{Muan2012} focused only on the regime with finite and fixed number of classes, and did not characterize the decay exponent of the error probability, whereas our focus is mainly on the asymptotic regime with infinite number of classes, and on the scaling behavior of the number of classes under which an asymptotically small error probability can be guaranteed. Nevertheless, the kernel-based approach developed in \cite{Muan2012} as well as in various other papers \cite{Srip2008,Fuku2009,Gretton2012} provide important techniques that we exploit in our study.

\textbf{Information theory in learning:} Quite a few recent studies have applied various notions in information theory for studying supervised learning problems.  A minimax approach for supervised learning, where the goal is to minimize the worst-case expected loss function over a certain set of probability distributions was developed in \cite{Farnia2016}. The designed classification rules are expected to be robust over datasets generated by any probability distribution in the set. 
A classification problem, where the observation is obtained via a linear mapping of a vector input was studied in  \cite{Nok2014}. The notion of classification capacity was proposed, which is similar to the discrimination capacity we propose. However, the results in \cite{Nok2014} are derived under the Gaussian model, whereas our formulation does not assume any specific distributions and is hence much more general. Furthermore, a parametric setting is implicitly assumed in \cite{Nok2014}, whereas our focus is on the nonparametric problem. {A connection between the hypothesis testing problems and channel coding was established in \cite{Nok2014}, compared to which this paper focuses primarily on the asymptotic case where the number of classes can scale.}
A supervised learning problem, where the joint distribution of the data sample and its label is assumed to be known but with an unknown parameter, was studied in \cite{Nok2016}. A classifier was proposed and the corresponding performance was analyzed. The connection of the problem to rate-distortion theory was explored. There are several key differences between the work in \cite{Nok2016} and our study. There is no notion of discrimination rate in \cite{Nok2016}, and the performance is not defined in terms of the asymptotic classification error probability. Additionally, our study does not assume any joint distribution of both the data sample and its label.

%% file: problem_formulation.tex
\section{Problem Formulation}
In this section, we first describe our composite nonparametric hypothesis model, and then connect it to the channel coding problem, which motivates several information theory related definitions that we will use to characterize system performance. For ease of readability, we also give preliminaries on the parametric hypothesis testing problem.

\subsection{Supervised Learning as Nonparametric Hypothesis Testing}\label{sec:nonparametric}
Consider the following nonparametric hypothesis testing problem with composite distributions. Suppose there are $M$ hypotheses, and each hypothesis corresponds to a set $\mathcal{P}_{m}$ of distributions for $m=1,\ldots,M$. For a given distance measure $d(p,q)$ between two probability distributions $p$ and $q$, we define
\begin{equation}
	\begin{aligned}
		d(\mathcal{P}_{m})& :=\sup_{p_{i},p_{i^{\prime}}\in\mathcal{P}_{m}}d(p_{i},p_{i^{\prime}}),\\
		d(\mathcal{P}_{m},\mathcal{P}_{m^{\prime}})&: =\inf_{p_{i}\in \mathcal{P}_{m},p_{i^{\prime}}\in\mathcal{P}_{m^{\prime}}}d(p_{i},p_{i^{\prime}}) \quad \text{for } m\neq m^{\prime}.
	\end{aligned}
\end{equation}
Hence, $d(\mathcal{P}_{m})$ represents the diameter of the $m$-th distribution set and $d(\mathcal{P}_{m},\mathcal{P}_{m^{\prime}})$ represents the inter-set distance between the $m$th and the $m^\prime$th sets.

We assume that
{
\begin{equation}\label{eq:KSassumptionHTS}
	\begin{aligned}
		\limsup_{ M \rightarrow \infty} \sup_{\substack{m=1,\ldots,M} }d(\mathcal{P}_{m})& < D_{I},\\
		\liminf_{M \rightarrow \infty}\inf_{\substack{m, m^\prime=1,\ldots,M\\m \neq m^{\prime}}}d(\mathcal{P}_{m},\mathcal{P}_{m^{\prime}}) & > D_O,\\
		D_{I} & < D_O.
	\end{aligned}
\end{equation}
}

That is, the intra-set distance (diameter) is always smaller than the inter-set distance for the composite hypothesis testing problem. {The actual values of $D_I$ and $D_O$ depend on the distance metrics used and are different.} { Furthermore, $\limsup\limits_{ M \rightarrow \infty}$ and $\liminf\limits_{ M \rightarrow \infty}$ in \eqref{eq:KSassumptionHTS} require that the conditions hold in the limit of asymptotically large $M$, i.e., the limit taken over the sequences of distribution clusters.}  We study the case where none of the distributions in the sets $\mathcal{P}_{m}$ for $m=1,\ldots,M$ are known. Instead, for $m=1,\ldots,M$, we assume that each distribution $p_{m,i_{m}}\in\mathcal{P}_{m}$, where $i_m \in I_1^{M_m} = \{1,2,\ldots,M_m\}$ is the index of the distribution, generates one training sequence {$\mathbf{x}_{m,i_{m}} \in \mathbb{R}^n$ consisting of {$n$} independently and identically distributed (i.i.d.)\ scalar training samples.} We use $\mathbf{X}_{m}$ to denote all training sequences generated by the distributions in $\mathcal{P}_{m}$. We assume that a test sequence {$\mathbf{y}\in \mathbb{R}^n$ of $n$ i.i.d.\ scalar samples} is generated by one of the distributions in one of the sets $\mathcal{P}_{m}$ for $m=1,\ldots,M$. The goal is to determine the hypothesis that the test sequence $\mathbf{y}$ belongs to, i.e., which set contains the distribution  that generated $\mathbf{y}$. {Note that since $M$ can scale with the number $n$ of samples (as we describe in the sequel), the assumption (2) should hold in the asymptotical regime as $n \rightarrow \infty$.}

{A practical example of the considered problem involves nonparametric detection of micro-Doppler modulated radar returns, such as those which occur in a ground moving target indicator (GMTI) radar \cite{5545175}. The micro-Doppler motion of a particular target generates a specific sideband structure, which varies within a distributional radius as the fundamental frequency of the target's micro motion changes, i.e., $D_I$. The difference between the fundamental sideband structure of the micro-Doppler modulations between different target types implies a distributional difference, i.e., $D_O$. This type of classification problem is clearly composite (based on an unknown fundamental modulation frequency), and a parametric realization is in many cases impractical as the specific physics of the movement can be very difficult to model in a closed form.}


Let $\delta(\{\mathbf{X}_{m}\}_{m=1}^{M},\mathbf{y})$ denote a test based on the given data. Then, the error probability for $\delta$ is defined as
\begin{align}\label{eq:classerrorprobHTS}
	{P}_{e}=&\sum_{m_{0}=1}^M {P}\bigg(\delta(\{\mathbf{X}_{m}\}_{m=1}^{M},\mathbf{y}) \neq m_{0}\big|\mathbf{y}\sim p_{m_{0},j}\in\mathcal{P}_{m_{0}}\bigg)\nn\\
&\cdot	{P}(m_{0}),
\end{align}
where ${P}(m_{0})$ is the \emph{a priori} probability that $\mathbf{y}$ is drawn from the $m_{0}$-th set of distributions.

For the above $M$-ary hypothesis testing problem, we are interested in the regime, in which the number $M$ of hypotheses scales with the number of samples. In particular, we assume $M=2^{nD}$, where the parameter $D$ captures how fast $M$ scales with $n$. We refer to $D$ as the {\em discrimination rate}.
\begin{definition}
We say that the discrimination rate $D$ is achievable, if there exists a classification rule $\delta$ for the multi-hypothesis testing problem such that the probability of error converges to zero as the number $n$ of observation samples converges to infinity.
\end{definition}
For a given composite hypothesis testing problem, we define the largest possible discrimination rate, $D$, to be the {\em discrimination capacity}, and denote it as $\barD$.

\subsection{Connection to the Channel Coding Problem}

Next, we discuss the connection between the asymptotic regime of the hypothesis testing problem and the channel coding problem studied in communications, which in fact motivated our definition of the discrimination rate and the discrimination capacity. 

In the channel coding problem (see Figure \ref{fig:channel}), assume there are $\cM=\{1,\ldots,2^{nR}\}$ messages to be transmitted with equal probability. An encoder maps each message $m\in \cM$ one-to-one onto a length-$n$ codeword $y_m^n=\{y_{m1},\ldots,y_{mn}\}$, which is transmitted over the channel. The channel maps each input symbol to an output symbol in a discrete memoryless fashion with the transition probability $P_{X|Y}(x|y)$ for each channel use, and the corresponding output sequence is given by $x^n=\{x_1,\ldots,x_n\}$. A decoder then estimates the original message as $\hat{m}$ based on the output sequence. 
Essentially, in the channel coding problem, there are a total of $M$ possible conditional distributions $p_m(x^n)=P_{X|Y}(x^n|y_m^n)$ given $y_m^n$, where $m=1,\ldots,M$, and the decoder determines which distribution $p^* \in \{ p_1,\ldots,p_{M} \}$ most probably generated the observed channel output $x^n$. 

The decoding process of the channel coding problem described above is a hypothesis testing problem. Inspired by the channel coding problem, our total number of hypotheses corresponds to the total number of messages in channel coding, and the discrimination rate $D$ we define corresponds to the communication rate $R$ in channel coding, which represents the transmitted message bits per coded symbol. By analogy, the discrimination rate $D$ can be interpreted as the number of class-bits that can be distinguished per observation sample. Similarly, the discrimination capacity $\bar{D}$ corresponds to the capacity in channel coding, and serves as the fundamental testing limit in hypothesis testing problems. {Note that in channel coding, the transmitter can choose to shape the distributions of transmitted symbols. Here, the hypothesis testing problem corresponds to the case where the distributions remain unshaped.}


Essentially, Shannon's channel coding theorem guarantees error-free transmission of an exponentially increasing number of messages provided that the transmission rate $R$ is less than the channel capacity $C$. In other words, Shannon's theorem implies that codewords $\{y^n\}$ can be designed such that  exponentially increasing number of conditional probability distributions can be distinguished given the channel output. Here, for the hypothesis testing problem, channel coding motivates us to investigate the following problems: 
\begin{itemize}
	\item Which tests distinguish an exponentially increasing number of hypotheses with asymptotically small error probability based on $n$ observation samples?
	\item What are the corresponding discrimination rates?
\end{itemize}

\begin{figure*}[t]
	\normalsize
	\centering
	\subfigure[An illustration of the channel coding problem.]{
		\includegraphics[width=.8\textwidth]{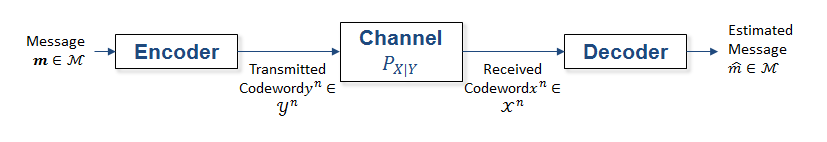}
		\label{fig:channel}
	} \\
	\subfigure[An illustration of the multiple hypothesis testing problem.]{
		\includegraphics[width=.7\textwidth]{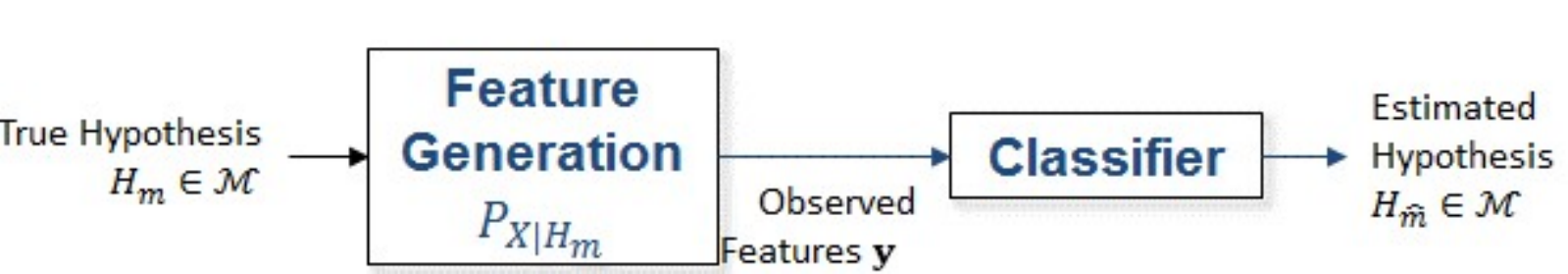}
		\label{fig:detection}
	}
	\caption{Illustrations of the channel coding problem and the multiple hypothesis testing problem}
	\label{figD}
\end{figure*}

%


\subsection{Preliminaries on Parametric Hypothesis Testing}\label{sec:para}

The aforementioned questions can be answered for the {\em parametric} hypothesis testing problem in the asymptotic regime based on existing studies, e.g., \cite{Cover06}. We explain this in detail for single distributions below as preliminary material before we delve into the main focus of this paper on the {\em nonparametric composite hypothesis testing problem}.


Consider the {\em parametric} hypothesis testing problem, where there are $M=2^{nD}$ {\em known} distinct distributions $p_1,...,p_M$ corresponding respectively to $M$ hypotheses. Given a test sequence {$\y$} consisting of $n$ i.i.d.\ samples generated from one of these distributions, the goal is to determine which hypothesis is true, i.e., which distribution $p_i$ generated the test sequence. 

We apply the likelihood test given by:
{
\begin{flalign}\label{eq:test0}
\delta(\y)=\arg\max_i P_{X|H_i}(\y) 
\end{flalign}}
where the test labels the observed test data as hypothesis $i$ if $p_i$ generates $\y$ with the largest probability. 
%
It can be shown \cite{Cover06} that the likelihood test in \eqref{eq:test0} is equivalently given by
{
\begin{flalign}\label{eq:test}
\delta(\y)=\arg\min_i D_{KL}(\gamma(\y)\|p_i),
\end{flalign}}
where $D_{KL}(\cdot\|\cdot)$ is the KLD between two distributions, and $\gamma(\cdot)$ is the empirical distribution of the sequence.
It suggests that the testing rule labels the test data as hypothesis $i$ if the empirical distribution of the test data is closest to $p_i$ in terms of KLD.

We next analyze the average error probability of the above testing rule as follows.
{
\begin{flalign}
P_e & =\frac{1}{M} \sum_{j=1}^{2^{nD}} P(\delta(\y) \neq j | H_j) \nn \\
&=\frac{1}{M} \sum_{j=1}^{2^{nD}} P(\exists\text{ $i\neq j$ s.t. } \mathcal{E}_1 |H_j) \nn \\
&\le \frac{1}{M} \sum_{j=1}^{2^{nD}} \sum_{i, i\neq j} P(\mathcal{E}_1 |H_j) \nn \\
& = \frac{1}{M} \sum_{j=1}^{2^{nD}} \sum_{i, i\neq j} \exp \{ -nC(p_i,p_j)\}\nn \\
&{ \leq 2^{nD-n\log e \liminf\limits_{M\rightarrow \infty}\min\limits_{1\leq i,j\leq M} C(p_i,p_j) },} \label{eq:error}
\end{flalign}
where $C(p_i,p_j)$ denotes the Chernoff distance}
\begin{align}
C(p_i,p_j)= \max_{0 \le t \le 1}-\log \int [p_i(p)]^{1-t}[p_j(p)]^tdp,
\end{align}
and $\mathcal{E}_1$ denotes the event that given $H_j$, the KLD between {$\y$} and $p_j$ is greater than the KLD between {$\y$} and $p_i$ for some $i\neq j$, i.e., for $i\neq j$,
{
\begin{align}
	D_{KL}(\gamma(\y)\|p_i)<D_{KL}(\gamma(\y)\|p_j).
\end{align}}
\noindent {Note that for simplicity, the default base for $\log$ in this paper is 2.} Thus, if { $D \leq \log e \liminf\limits_{M\rightarrow \infty}\min\limits_{1\leq i,j\leq M} C(p_i,p_j) $}, then the error probability is asymptotically small as $n$ goes to infinity, which proves the following proposition. 
\begin{proposition}
	For the parametric multiple hypothesis testing problem, the discrimination rate $D$ is achievable if 
	\[ { D \leq \log e \liminf\limits_{M\rightarrow \infty}\min\limits_{1\leq i,j\leq M} C(p_i,p_j).}\]
\end{proposition}
{ Hence, for the discrimination rate to be positive, we require that the smallest pairwise Chernoff information be bounded away from zero for asymptotically large $M$, i.e., the limit taken over the sequences of distribution clusters.}

%% file: analysis_results.tex
\section{Main Results}
In this section, we obtain the performance bounds for the nonparametric hypothesis testing problem, with two different distance measures, i.e., MMD and KS distance.
\subsection{MMD-Based Test}\label{sec:MMDresults}
We construct a nonparametric hypothesis test based on the MMD distance between two distributions $p$ and $q$ \cite{Gretton2012} defined as follows
\begin{equation}
\mmd^2(p,q):=\|\mu_p-\mu_q\|_{\cH}.
\end{equation}
where $\mu_p(\cdot)$ maps a distribution $p$ into an element in a reproducing kernel Hilbert space (RKHS) associated with a kernel $k(\cdot,\cdot)$ as
\begin{align}
	\mu_p(\cdot)=\mE_p [k(\cdot,x)]=\int k(\cdot,x)dp(x).
\end{align}
{An unbiased estimator of $\mmd^2[p,q]$ based on $n$ samples of $\x=\{x_1,\ldots,x_n\}$ generated by distribution $p$ and $m$ samples of $\y=\{y_1,\ldots,y_m\}$ generated by distribution $q$, is given by \cite{Gretton2012}:
\begin{flalign}\label{eq:mmdu}
&\mmd^2(\x,\y)=\frac{1}{n(n-1)}\sum_{i=1}^n\sum_{j\neq i}^n k(x_i,x_j)\nn \\
& +\frac{1}{m(m-1)}\sum_{i=1}^m\sum_{j\neq i}^m k(y_i,y_j)-\frac{2}{nm}\sum_{i=1}^n\sum_{j=1}^m k(x_i,y_j).
\end{flalign}
Note that $x_i,y_i \in \mathbb R^d$, and the dimension $d \ge 1$.}

We employ the MMD to measure the distance between the test sequence and the training sequences, and declare the hypothesis of the test sequence to be the same as the training sequence that has the smallest MMD to the test sequence. The constructed MMD-based nonparametric composite hypothesis test is given by
\begin{align}\label{mmd_test}
\delta_{\mmd}(\{\mathbf{X}_{m}\}_{m=1}^{M},\mathbf{y})=\arg\min_{m,i_m} \mmd^2(\mathbf{x}_{m,i_{m}},\mathbf{y}).
\end{align}

 The following theorem characterizes the average probability of error performance of the proposed MMD-based test under composite distributions.
\begin{theorem}\label{comp_mmd}
	Suppose the MMD-based test in \eqref{mmd_test} is applied to the nonparametric composite hypothesis testing problem under assumption \eqref{eq:KSassumptionHTS}, where the kernel satisfies $0\leq k(x,y)\leq \K$ for all $(x,y)$. Then, the average probability of error is upper bounded as{
	\begin{equation*}
	{P}_{e} \leq  2^{nD}\exp \left(-\frac{n(D_O-D_I)^2}{{96\K^2}}\right).
	\end{equation*}}
	Thus, the achievable discrimination rate is
	{
	\begin{align}
	D = \frac{ \log e}{96\K^2} (D_O-D_I)^2.
	\end{align}
	}
\end{theorem}
\begin{proof}
	See Appendix \ref{proof_comp_mmd}.
\end{proof}
Next, we study a special case where each hypothesis is associated with a single distribution, i.e., the $m$-th hypothesis is associated with only one distribution  $p_{m}$, $m=1,\ldots, M$. Then, we have the following corollary.
\begin{corollary}
\label{mmd}
	Suppose the MMD-based test is applied to the nonparametric hypothesis testing problem {under assumption \eqref{eq:KSassumptionHTS}}, and each hypothesis is associated with a single distribution, where the kernel satisfies $0\leq k(x,y)\leq \K$ for all $(x,y)$. Then, the average probability of error under equally probable hypotheses is upper bounded as
	{
	\begin{align}
	P_e  \leq 2^{nD-n\frac{\log e}{96\K^2} { \liminf\limits_{M\rightarrow \infty}\min\limits_{1\leq i,j\leq M}} \text{MMD}^4(p_i,p_j)}.
	\end{align}}
	Thus, the achievable discrimination rate is
	{
	\begin{align}
	D = \frac{\log e}{96\K^2} {\liminf\limits_{M\rightarrow \infty}\min\limits_{1\leq i,j\leq M}} \text{MMD}^4(p_i,p_j).
	\end{align}}
\end{corollary}
{Note that, for the discrimination rate to be positive, we require the smallest pairwise MMD between the distributions to be bounded away from zero for asymptotically large $M$, where the limit is taken over the sequences of distribution clusters.}
\begin{proof}
	By Theorem \ref{comp_mmd}, we set $D_I=0$ and $D_O={\liminf\limits_{M\rightarrow \infty}\min\limits_{1\leq i,j\leq M}} \mmd^2(p_i,p_j)$.
	\noindent Therefore, we can bound the probability of error as the number of classes scales according to $M=2^{nD}${
	\begin{align}
		P_e & \le M \exp \left(-\frac{n{ \liminf\limits_{M\rightarrow \infty}\min\limits_{1\leq i,j\leq M}}\text{MMD}^4(p_i,p_j)}{96\K^2}\right) \nn \\
		&\leq 2^{nD-n\frac{\log e}{96\K^2} {\liminf\limits_{M\rightarrow \infty}\min\limits_{1\leq i,j\leq M}}\mmd^4(p_i,p_j)}.
	\end{align}}
	
	\noindent Then, it is straightforward to obtain thebuyong  achievable discrimination rate for the MMD test as
	{
	\begin{align}
		D = \frac{\log e}{96\K^2}{\liminf\limits_{M\rightarrow \infty}\min\limits_{1\leq i,j\leq M}} \text{MMD}^4(p_i,p_j).
	\end{align}}
\end{proof}

\subsection{Kolmogorov-Smirnov Test}\label{sec:KSresults}
				
				
				In this section, we construct a nonparametric hypothesis testing test based on the KS distance defined as follows. Suppose $\mathbf{x}=\{x_{1},\ldots,x_{n}\}$, and i.i.d. samples $x_i\in \mathbb{R},$ are generated by the distribution $p$. Then the empirical CDF of $p$ is given by
				{\begin{equation}
				F_{\mathbf{x}}(a)=\frac{1}{n}\sum_{i=1}^{n}1_{[-\infty,a]}(x_{i}),
				\end{equation}}
				where $1_{[-\infty,x]}$ is the indicator function. The KS distance between $\mathbf{x}$ and $\mathbf{y}$ having respectively been generated by $p$ and $q$ is defined as
				{
				\begin{equation}
				D_{KS}(\mathbf{x},\mathbf{y}) = \sup_{a\in\mathbb{R}}|F_{\mathbf{x}}(a)- F_{\mathbf{y}}(a)|.
				\end{equation}}
				
				We construct the following KS based nonparametric composite hypothesis test
				\begin{align}\label{ks_test}
				\delta_{KS}(\{\mathbf{X}_{m}\}_{m=1}^{M},\mathbf{y})=\arg\min_{m,i_m} D_{KS}(\mathbf{x}_{m,i_{m}},\mathbf{y}),
				\end{align}

				The following theorem characterizes the performance of the proposed KS-based test.
\begin{theorem}\label{comp_ks}
	Suppose the KS-based test in \eqref{ks_test} is applied to the nonparametric hypothesis testing problem under assumption \eqref{eq:KSassumptionHTS}. Then, the average probability of error is upper bounded as
	{
	\begin{equation*}
	{P}_{e} \leq 6\cdot 2^{nD}\exp\big(-\frac{n(D_O-D_{I})^{2}}{8}\big).
	\end{equation*}}
	Thus, the achievable discrimination rate is
	{
	\begin{align}
	D = \frac{\log e}{8} (D_O-D_I)^2.
	\end{align}}
\end{theorem}
\begin{proof}
	See Appendix \ref{proof_comp_ks}.
\end{proof}

We next consider the case where each hypothesis is associated with a single distribution.
Then, we have the following corollary.
\begin{corollary}\label{ks}
			Suppose the KS-based test is applied to the nonparametric hypothesis testing problem { under assumption \eqref{eq:KSassumptionHTS}}, and each hypothesis is associated with a single distribution. Then, the average probability of error under equally probable hypotheses is upper bounded as{
				\begin{align}
				P_e  \leq 6\cdot 2^{nD-n\frac{\log e }{8}{\liminf\limits_{M\rightarrow \infty}\min\limits_{1\leq i,j\leq M}}d_{KS}^{2}(p_i,p_j)}.
				\end{align}}
				Thus, the achievable discrimination rate is{
				\begin{align}
				D = \frac{\log e }{8}{ \liminf\limits_{M\rightarrow \infty}\min\limits_{1\leq i,j\leq M}}d_{KS}^{2}(p_i,p_j).
				\end{align}}
\end{corollary}
{ Hence, for the discrimination rate to be positive, we require the least pairwise KS distance between distributions to be bounded away from zero for asymptotically large $M$, where the limit is taken over the sequences of distribution clusters.}
		\begin{proof}
	By Theorem \ref{comp_ks}, we set $D_{I} = 0$ and $D_O={\liminf\limits_{M\rightarrow \infty}\min\limits_{1\leq i,j\leq M}}d_{KS}(p_{i},p_{j})$, and have
	{
	\begin{equation*}
		\begin{aligned}
			P_{e} & \leq 6\cdot 2^{nD-n\frac{\log e}{8}D_{o}^{2}}\\
			&\leq 6\cdot 2^{nD-n\frac{\log e}{8}{ \liminf\limits_{M\rightarrow \infty}\min\limits_{1\leq i,j\leq M}}d_{KS}^2(p_{i},p_{j})}.
		\end{aligned}
	\end{equation*}}
	
	\noindent Then, it is straightforward to obtain the following  achievable discrimination rate for the KS test as
	{
	\begin{align}
		D = \frac{\log e }{8}{\liminf\limits_{M\rightarrow \infty}\min\limits_{1\leq i,j\leq M}}d_{KS}^{2}(p_i,p_j).
	\end{align}}
		\end{proof}
		
\subsection{Upper Bound on the Discrimination Capacity}\label{sec:fanoResults_comp}

In this section, we provide an upper bound on the discrimination capacity for the composite hypothesis testing problem. Let $h$ be a random index representing the actual hypothesis that occurs. We assume that $h$ is uniformly distributed over the $M$ hypotheses, and $h^\prime$ has the same distribution as $m$, but is independent from $h$. { Then, Lemma 2.10 in \cite{Tsybakov:2008:INE:1522486} directly yields the following upper bound on the discrimination capacity $\bar{D}$. }
\begin{remark}\label{dis_capacity_comp}
	The discrimination capacity $\bar{D}$ is upper bounded as
	\begin{align}
	\bar{D} \leq { \limsup_{M\rightarrow \infty}} \mE_{h,h^\prime} D_{KL}(p_h\|p_{h^\prime}),
	\end{align}
	where $D_{KL}(\cdot\|\cdot)$ is the KLD between two distributions.
\end{remark}
{ Note that the above limit $\limsup_{M\rightarrow \infty}$ is taken over the sequences of distribution clusters.}

{ In Appendix \ref{proof_capacity_comp}, we provide an alternative but simpler proof based on Fano's inequality for the above upper bound, which is closely related to the proposed concept of discrimination capacity.}
%

{
\subsection{Training Sequences of Unequal Length}\label{sec:varyinglength}
In this subsection, we discuss the impact of different number of training samples in different classes on the probability of error and the discrimination rate. Here, we still assume that there are $n$ test samples. To keep the problem formulation meaningful, we assume that the number $M$ of classes increases exponentially with $n$ at a rate $D$, i.e., $M=2^{nD}$.  To avoid notational confusion, we use the non-composite case, i.e., with each class corresponding to one distribution, to illustrate the idea. Suppose that each class, i.e., each distribution, generates $\gamma_m(n)$ training samples, for $m=1, \ldots,M$, where $\gamma_m(n)$ represents the number of samples in the $m$-th class (as a function of $n$). Let $\gamma_{\min}(n)=\min_{1\leq m\leq M} \gamma_m(n)$. 
In particular, for the MMD-based test, the probability of error can be bounded as
	\begin{align}\label{Pe-multi-length-MMD}
	P_e \le 2^{n\left(D-\min\{1,\frac{\gamma_{\min}(n)}{n}\}\frac{\log e (D_O-D_I)^2}{96\mathcal{K}^2}\right)}.
	\end{align}
	For the KS-based test, the probability of error can be bounded as
	\begin{align}\label{Pe-multi-length-KS}
	P_e \le 6\cdot 2^{n\left(D-\min\{1,\frac{\gamma_{\min}(n)}{n}\}\frac{\log e (D_O-D_I)^2}{8}\right)}.
	\end{align} 
It can be seen that here the ratio $\frac{\gamma_{\min}(n)}{n}$ plays an important role in determining the error exponent asymptotically. For example, for the MMD-based test, if the ratio converges to zero for large $n$, i.e., the shortest training length $\gamma_{\min}(n)$ scales as an order-level slower than the test length, then there is no guarantee of exponential error decay, and the discrimination rate equals zero. On the other hand, if $\lim_{n\rightarrow \infty} \frac{\gamma_{\min}(n)}{n} =c$ with $0 < c <1$, then the discrimination rate $D=c\frac{\log e (D_O-D_I)^2}{96\mathcal{K}^2}$. Furthermore, if $\lim_{n\rightarrow \infty} \frac{\gamma_{\min}(n)}{n} =c$ with $c\ge 1$, then the discrimination rate $D=\frac{\log e (D_O-D_I)^2}{96\mathcal{K}^2}$. A sketch of the proof of \eqref{Pe-multi-length-MMD} and \eqref{Pe-multi-length-KS} can be found in Appendix \ref{multi-length-proof-sketch}

}

%% file: simulation.tex
\section{Numerical Results}\label{sec:numResults}

In this section, we present numerical results to compare the performance of the proposed tests. In the experiment, the number of classes is set to be five, and the error probability versus the number of samples for the proposed algorithms is plotted. For the MMD based test, we use the standard Gaussian kernel given by $k(x,x^\prime)=\exp(-\frac{\|x-x^\prime\|^2}{2})$.
\begin{figure}[!htb]
\centering
\includegraphics[width=0.45\textwidth, height=2.5in]{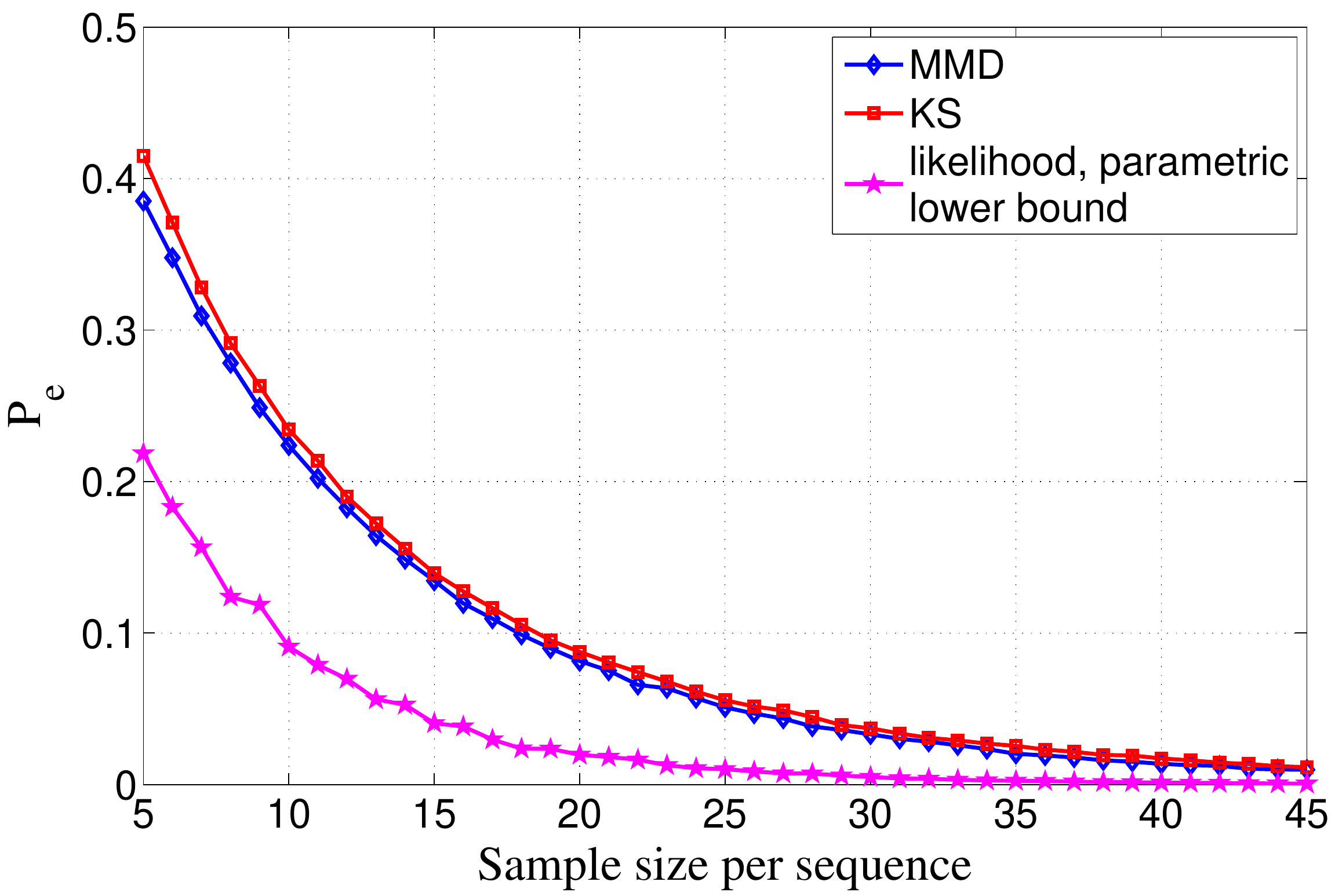}
\caption{Comparison of error probabilities of different hypothesis testing algorithms for Gaussian distributions with different means.}
\label{fig:meandiff}
\end{figure}

\begin{figure}[!htb]
\centering
\includegraphics[width=0.45\textwidth,height=2.5in]{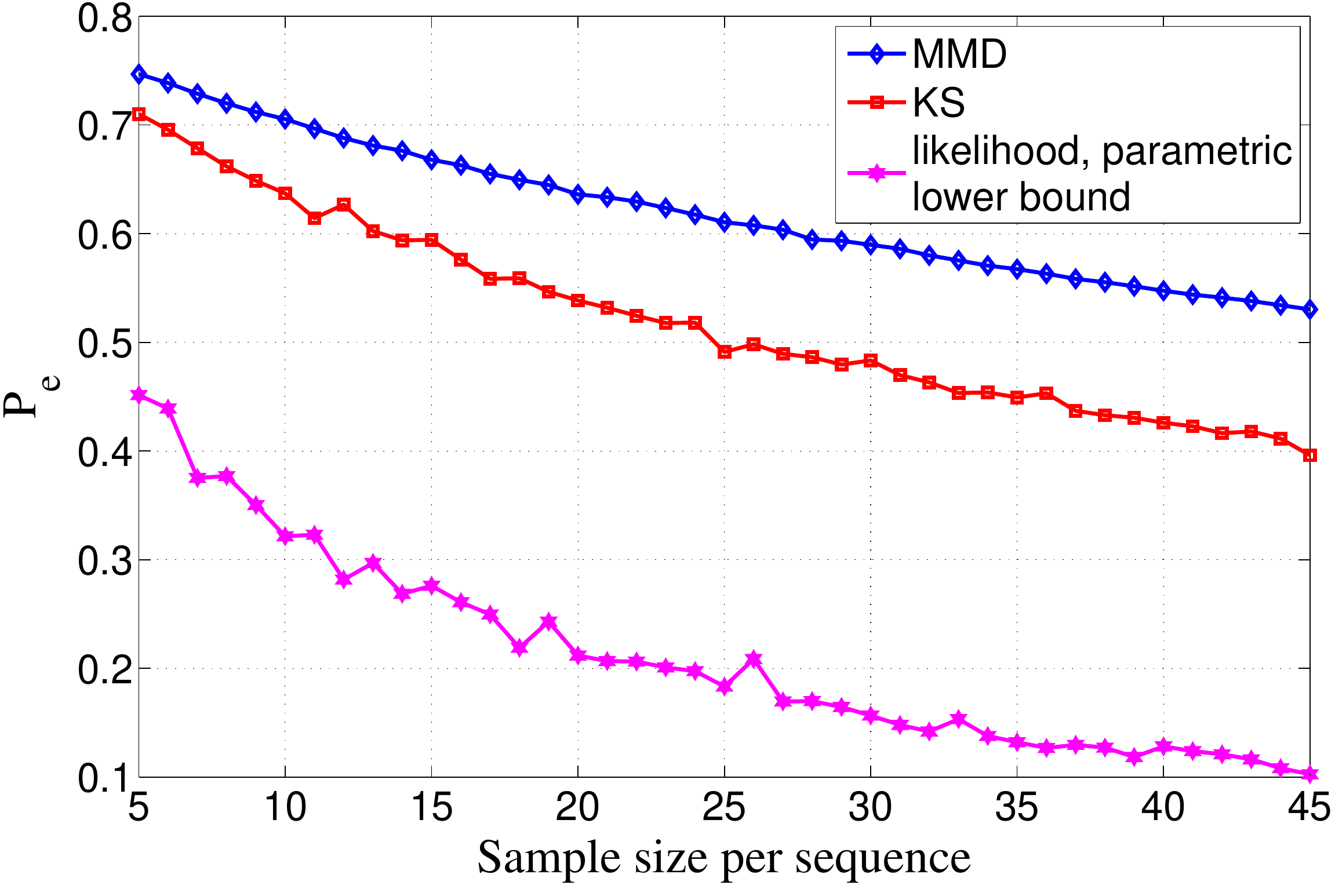}
\caption{Comparison of error probabilities of different hypothesis testing algorithms for Gaussian distributions with different variances.}
\label{fig:variancediff}
\end{figure}

In the first experiment, all the hypotheses correspond to Gaussian distributions with the same variance $\sigma^2=1$ but different mean values $\mu=\{-2,-1,0,1,2\}$. A training sequence is drawn from each distribution and a test sequence is randomly generated from one of the five distributions. The sample size of each sequence ranges from $5$ to $45$. A total of $10^{5}$ monte carlo runs are conducted. The simulation results are given in Figure~\ref{fig:meandiff}. It can be seen that all the tests give better performance as the sample size $n$ increases. We can also see that the MMD-based test slightly outperforms the KS-based test. We also provide results for the parametric likelihood test as a lower bound on the probability of error for performance comparison. It can be seen that the performance of the two nonparametric tests are close to the parametric likelihood test even with a moderate number of samples.

In the second experiment, all the hypotheses correspond to Gaussian distributions with the same mean $\mu=1$ but different variance values $\sigma^2 = \{0.5^2,1^2,1.5^2,2^2,2.5^2\}$. The simulation results are given in Fig.~\ref{fig:variancediff}. In this experiment, the MMD-based test yields the worst performance, which suggests that this method is not suitable when the distributions overlap substantially with each other. 
The two simulation results also suggest that none of the three tests perform the best universally over all distributions. Although there is a gap between the performance of MMD and KS tests and that of the parametric likelihood test, we observe that the error decay rates of these tests are still close.

{To show the tightness of the bounds derived in the paper, we provide a table (See Table I) of error decay exponents (and thus the discrimination rates) for different algorithms. 
\begin{table}[]
	\centering
	\caption{Comparison of Bounds}
	\label{my-label}
	\begin{tabular}{l|l|l|l|l|l}
		\cline{2-5}
		\multirow{2}{*}{}                  & \multicolumn{2}{l|}{Lower Bounds} & \multicolumn{2}{l|}{Upper Bounds}   &  \\ \cline{2-5}
		& KS         & MMD        & Parametric & FB           &  \\ \cline{1-5}
		\multicolumn{1}{|l|}{Empirical}  & 0.0897          & 0.0916          & 0.146      & 2.5                    &  \\ \cline{1-5}
		\multicolumn{1}{|l|}{Theoretical} & 0.0183          & 0.0071          & 0.125      & \multicolumn{1}{c|}{-} &  \\ \cline{1-5}
	\end{tabular}
\end{table}
\begin{figure}[!htb]
	\centering
	\includegraphics[width=0.5\textwidth,height=2.5in]{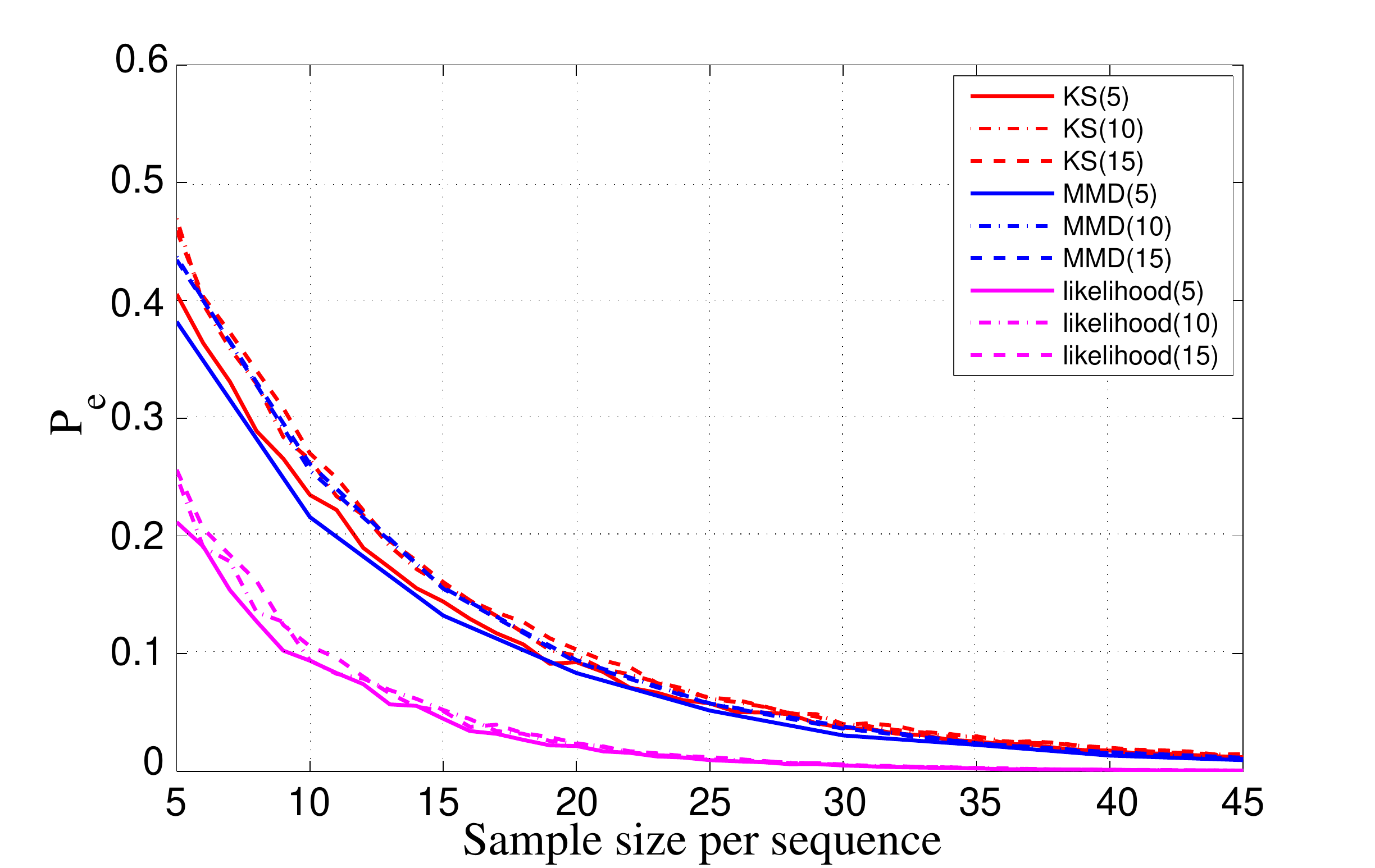}
	\caption{Comparison of error probabilities for different hypothesis testing algorithms for Gaussian distributions with different means.}
	\label{fig:meandiff_fortable}
\end{figure}
Estimates of error decay exponent  of KS and MMD based tests on a multi-hypothesis testing problem are presented for the problem considered in the first experiment. {Note that the theoretical lower bounds in the table correspond to the achievable discrimination rates of the methods asymptotically.} Fano's bound (FB in the table) is estimated by using data-dependent partition estimators of Kullback-Leibler divergence \cite{1499042}. The parametric upper bound is based on the maximum likelihood test, which can serve as an upper bound on the error decay exponent (and hence intuitively on the discrimination capacity). It can be seen from the table that both the KS and MMD tests do achieve an exponential error decay and have positive discrimination rates as we show in our theorems. Clearly, the empirical values of the bounds for both tests are better than the corresponding theoretical values
. More importantly, both of the empirical lower bounds are close to the likelihood upper bound, demonstrating that the actual performance of the two tests are satisfactory. We also note that the Fano's upper bound is not very close to the lower bound.

To better illustrate the bounds in Table I, we provide experimental results with different number of hypotheses $M$ in Figure 4. In particular, we present the simulation results with $M=5, 10, 15$. We use a similar experiment setting as that in the first experiment, where Gaussian distributions have the same variance and different mean values, and the mean values are $\{-2,-1,\ldots,2\}, \{-4.5,-3.5, \ldots, 4.5\}$ and $\{-7,-6,\ldots,7\}$ respectively. The parametric maximum likelihood test serves as an unpper bound for the error decay exponent for all of the three cases. Similar to the case $M=5$, KS and MMD nonparametric tests achieve an exponential error decay and hence the positive discrimination rates for the cases $M=10$ and $M=15$.}

We now conduct experiments with composite distributions. First, we still use five hypotheses with Gaussian distributions with variance  $\sigma^2=1$ and different mean values $\mu=\{-2,-1,0,1,2\}$. For each hypothesis, we vary the mean values by $\pm 0.1$. Thus, within each hypothesis, there are three different distributions with mean values in $\{\mu-0.1, \mu, \mu+0.1\}$. The results are presented in Figure \ref{fig:comp_meandiff}. As expected, the performance improves as the sample size $n$ increases. The two tests perform almost identically, with the MMD-based test slightly outperforming the KS-based test for small $n$. 

We again vary the variances of the Gaussian distributions as in the second experiment in a similar way. In particular, the variances in the same class are $\{(\sigma-0.1)^2,\sigma^2,(\sigma+0.1)^2\}$, and $\sigma \in  \{0.5,1,1.5,2,2.5\}$ . In Figure \ref{fig:comp_variancediff}, we observe the performance improvement as the sample size $n$ increases. Different from the results in the second experiment, the MMD-based test outperforms the KS-based test in the composite setting. 

\begin{figure}[!htb]
	\centering
	\includegraphics[width=0.45\textwidth,height=2.5in]{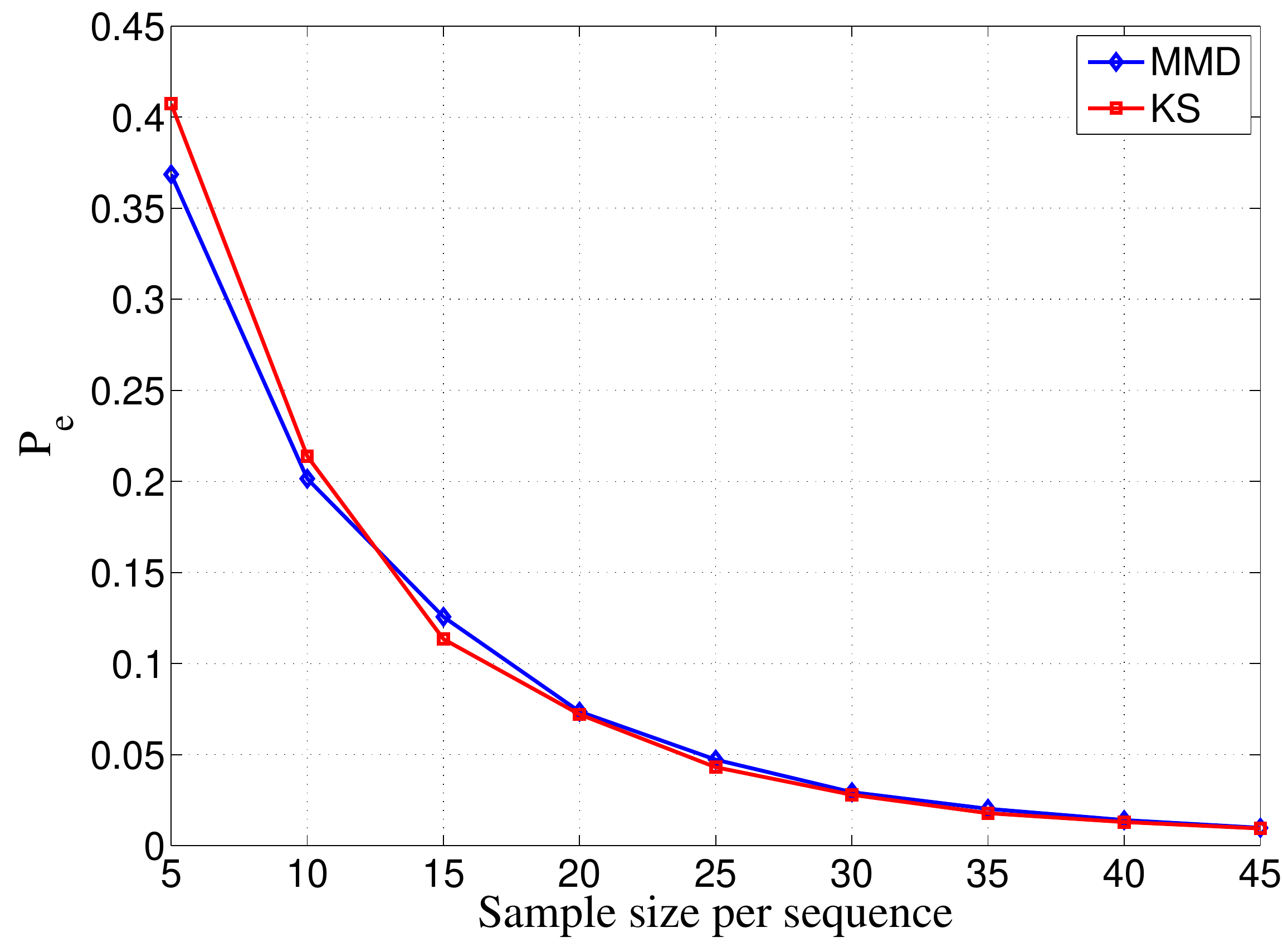}
	\caption{Comparison of error probabilities of different hypothesis testing algorithms for composite Gaussian distributions with different means.}
	\label{fig:comp_meandiff}
\end{figure}

\begin{figure}[!htb]
	\centering
	\includegraphics[width=0.45\textwidth,height=2.5in]{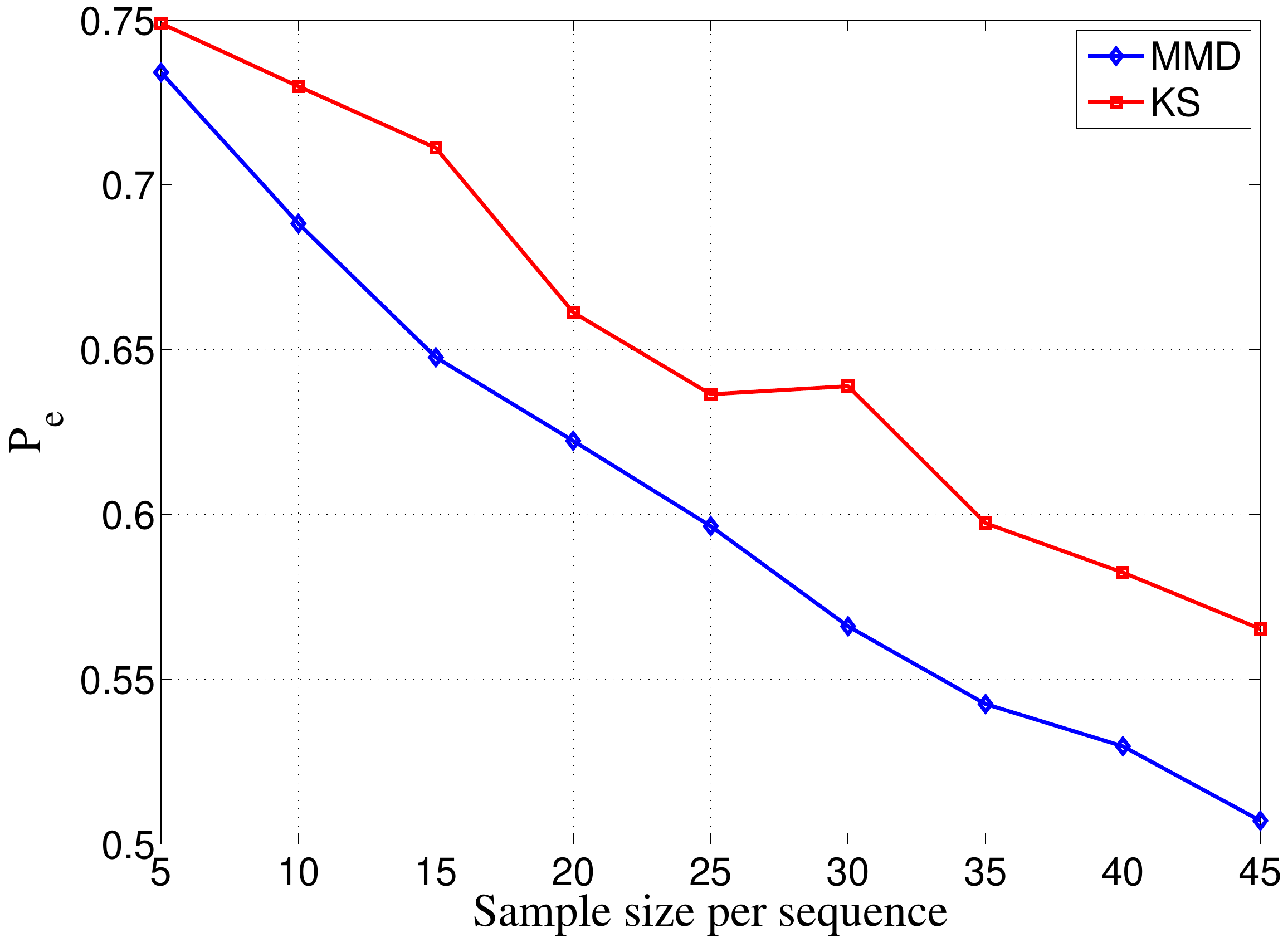}
	\caption{Comparison of error probabilities of different hypothesis testing algorithms for composite Gaussian distributions with different variances.}
	\label{fig:comp_variancediff}
\end{figure}

%% file: conclusion.tex
\section{Conclusion}
This paper developed a nonparametric composite hypothesis testing approach for arbitrary distributions based on the maximum mean discrepancy (MMD) and Kolmogorov-Smirnov (KS) distance measure based tests. We introduced the information theoretic notion of discrimination capacity that was defined for the regime where the number of hypotheses scales along with the sample size. We also provided characterization of the corresponding error exponent and the discrimination rate, i.e., a lower bound on the discrimination capacity. Our framework can been extended to unsupervised learning problems and similar performance limits can be investigated.

%% file: appendix.tex
\appendix
\subsection{Proof of Theorem \ref{comp_mmd}}\label{proof_comp_mmd}
		The proof uses the following inequality.
		\begin{lemma}{[McDiarmid's Inequality \cite{mcdiarmid1989method}]}\label{lemma:mcd}
			Let $f: \mathcal X^m \rightarrow \mathbb R$ be a function such that for all $i\in \{1,\ldots, m\}$, there exist $c_i \le \infty$ for which
			\begin{align}
			\sup \limits_{X\in \mathcal X^m , \tilde{x} \in \mathcal{X} } |g(x_1,\ldots,x_{i-1},\tilde{x}, x_{i+1},\ldots,x_m)|\le c_i,
			\end{align}
			where $g(x_1,\ldots,x_{i-1},\tilde{x}, x_{i+1},\ldots,x_m)=f(x_1,\ldots, x_m)-f(x_1,\ldots,x_{i-1},\tilde{x}, x_{i+1},\ldots,x_m)$.
			Then for all probability measure $p$ and every $\epsilon>0$,
			\begin{align}
			P_X\left(f(X)-\mE_X[f(X)] >\epsilon\right)< \exp\left(-\frac{2\epsilon^2}{\sum_{i=1}^{m}c_i^2}\right),	
			\end{align}
			where $X$ denotes $(x_1,\ldots,x_m)$, $\mE_X[\cdot]$ denotes the expectation over the $m$ random variables $x_i \sim p$, and $P_X$ denotes the probability over these $m$ variables.
		\end{lemma}
		
To apply the McDiarmid's inequality, we first define the following quantity
\begin{align}
\bigtriangleup_m(\x^\alpha)=\text{MMD}^2(\x_{m,i_m},\y) - \text{MMD}^2(\x_{m^\prime,i_{m^\prime}},\y)
\end{align}
where $\x^\alpha:=\{\x_{m,i_m},\x_{m^\prime,i_{m^\prime}},\y\}$ consists of 
$3n$ data samples.

\noindent Given $H_i$, i.e., the test sequence $\y$ is generated by $\mathcal{P}_i$, it can be shown that
\begin{align}\label{di}
\mathbb{E}[\text{MMD}^2(\x_{m,i_m},\y)] \le D_I
\end{align}
and
\begin{align}\label{do}
\mathbb{E}[\text{MMD}^2(\x_{m^\prime,i_{m^\prime}},\y)]\ge D_O.
\end{align}

We next define $\x^\alpha_{-s}$ the same as $\x^\alpha$ except that the $s$-th component $\x^\alpha_s$ is removed. We also define $\tilde{\x}^\alpha_s$ as another sequence generated by the same underlying distribution for ${\x}^\alpha_s$. Then, $\x_s^\alpha$ affects $\bigtriangleup_m(\x^\alpha)$ via the following three cases.
\begin{itemize}
	\item Case 1: $\x^\alpha_s$ is in the sequence $\x_{m,i_m}$. In this case, $\x^\alpha_s$ affects $\bigtriangleup_m(\x^\alpha)$ through the following terms
	{
	\begin{align}
	\frac{2}{n(n-1)}\sum\limits_{l=1,l\ne s}^n k(\x_s^\alpha,\x_{m,i_m}(l))-\frac{2}{n^2}\sum\limits_{l=1}^n k(\x^\alpha_s, \y(l)).
	\end{align}}
	\item Case 2: $\x^\alpha_s$ is in the sequence $\x_{m^\prime,i_{m^\prime}}$. In this case, $\x^\alpha_s$ affects $\bigtriangleup_m(\x^\alpha)$ through the following terms
	{
	\begin{align}
	\frac{2}{n(n-1)}\sum\limits_{l=1,l\ne s}^n k(\x^\alpha_s, \y(l))
	-	\frac{2}{n^2}\sum\limits_{j=1}^n k(\x_s^\alpha,\x_{m^\prime,i_{m^\prime}}(l)).
	\end{align}}
	
	\item Case 3: $\x^\alpha_s$ is in the sequence $\y$. In this case, $\x^\alpha_s$ affects $\bigtriangleup_m(\x^\alpha)$ through the following terms
	{
	\begin{align}
	\frac{2}{n^2}\sum\limits_{l=1}^n k(\x_s^\alpha,\x_{m,i_m}(l))
	-\frac{2}{n^2}\sum\limits_{l=1}^n k(\x^\alpha_s, \x_{m^\prime,i_{m^\prime}}(l)).
	\end{align}}
\end{itemize}	
Thus, since the kernel is bounded, i.e., $0 \leq k(x; y) \leq \K$ for any $(x, y)$,  considering the above three cases, the variation in the value of $\bigtriangleup_m(\x^\alpha)$ when $\x^\alpha_s$ varies is bounded by {$\frac{4\K}{n}$.} Then, we can obtain the following bound
\begin{align}\label{mmd-seq-bound}
|\bigtriangleup_m(\x^\alpha_{-s}, \x^\alpha_s)- \bigtriangleup_m(\x^\alpha_{-s}, \tilde{\x}^\alpha_s)| \le \frac{8\K}{n}.
\end{align}

\noindent Assuming the test sequence is generated from the true underlying set of distributions, $\mathcal P_m$, we now apply Lemma \ref{lemma:mcd} and obtain the following bound on the probability of error between two classes
{
\begin{align}\label{Pe-same-length-mmd}
& P\left(\text{MMD}^2(\x_{m,i_m},\y) \ge \text{MMD}^2(\x_{m^\prime,i_{m^\prime}},\y)\right) \nn\\
&=  P(\text{MMD}^2(\x_{m,i_m},\y)- \text{MMD}^2(\x_{m^\prime,i_{m^\prime}},\y) \geq 0)\nn \\
 &=	P\left(\bigtriangleup_m(\x^\alpha)-\mathbb{E}[\bigtriangleup_m(\x^\alpha)]\ge -\mathbb{E}[\bigtriangleup_m(\x^\alpha)] \right) \nn \\
&\le	 P\left(\bigtriangleup_m(\x^\alpha)-\mathbb{E}[\bigtriangleup_m(\x^\alpha)]\ge D_O-D_I \right) \nn \\
& \le \exp \left(-\frac{2(D_O-D_I)^2}{64\K^2\frac{3}{n}}\right)\nn\\
&\le \exp \left(-\frac{n(D_O-D_I)^2}{{96\K^2}}\right).
\end{align}}
The first inequality is based on the results in \eqref{di} and \eqref{do} that
\begin{align}
	-\mathbb{E}[\bigtriangleup_m(\x^\alpha)] \ge D_O-D_I.
\end{align}
\noindent Therefore, we can bound the probability of error as
{
\begin{align}
 P_e & =P\left(\exists m^\prime \ne m, i_m^\prime \in I_1^{M_m^\prime},\bigtriangleup_m(\x^\alpha)\ge 0, \forall i_m \in I_1^{M_m}\right)\nn\\
& \le \frac{1}{M} \sum_{m=1}^{M}\sum_{m^{\prime}\neq m} \min_{i_m \in I_1^{M_m}} \exp \left(-\frac{n(D_O-D_I)^2}{{96\K^2}}\right)\nn \\
&\leq M\exp \left(-\frac{n(D_O-D_I)^2}{{96\K^2}}\right).
\end{align}}

Thus, the achievable discrimination rate is
{
\begin{align}
D = \frac{\log e}{96\K^2} (D_O-D_I)^2.
\end{align}}

\subsection{Proof of Theorem \ref{comp_ks}}\label{proof_comp_ks}	
We first introduce two lemmas to help establish the  theorem.
\begin{lemma}\cite{Massart1990}\label{lemma:KStest1}
	Suppose $\mathbf{x}$ is generated by $p$ and $F_{\mathbf{x}}(a)$ is the corresponding empirical c.d.f.. Then
	\begin{equation*}
	P\bigg(\sup_{a\in\mathbb{R}}\bigg|F_{\mathbf{x}}(a)-F_{p}(a)\bigg|>\epsilon\bigg)\leq 2\exp\big(-2n\epsilon^{2}\big).
	\end{equation*}
\end{lemma}
\begin{lemma}\label{lemma:KStest4}
	Suppose two distribution clusters $\mathcal{P}_{1}$ and $\mathcal{P}_{2}$ satisfy \eqref{eq:KSassumptionHTS}. Assume that for $j=1,2$, $\mathbf{x}_{j}\sim p_{j}$ satisfying $p_{j}\in \mathcal{P}_{j}$. Then for any $\mathbf{x}_{3}\sim p_{3}$ satisfying $p_{3}\in\mathcal{P}_{1}$,
	\begin{equation*}
	\begin{aligned}
	{P}\bigg(d_{KS}(\mathbf{x}_{1},\mathbf{x}_{3}) \geq d_{KS}(\mathbf{x}_{2},\mathbf{x}_{3})\bigg)\leq 6\exp{\big(-\frac{n(D_{O}-D_{I})^{2}}{8}\big)}.
	\end{aligned}
	\end{equation*}
\end{lemma}
\begin{proof}
	By the triangle inequality and the property of supremum, we have
	\begin{equation*}
	\begin{aligned}
	d_{KS}(\mathbf{x}_{1},\mathbf{x}_{3}) & < d_{KS}(p_{1},\mathbf{x}_{1}) +  d_{1} + d_{KS}(p_{3},\mathbf{x}_{3}),\\
	d_{KS}(\mathbf{x}_{2},\mathbf{x}_{3}) & > -d_{KS}(p_{3},\mathbf{x}_{3}) +  d_{2} - d_{KS}(p_{2},\mathbf{x}_{2}).
	\end{aligned}
	\end{equation*}
	where $D_{I}<d_{1}<d_{2}<D_{O}$. Then
	\begin{equation*}
	\begin{aligned}
	&\quad\:{P}\bigg(d_{KS}(\mathbf{x}_{1},\mathbf{x}_{3}) \geq d_{KS}(\mathbf{x}_{2},\mathbf{x}_{3})\bigg)\\
	& \leq {P}\bigg(d_{KS}(p_{1},\mathbf{x}_{1}) + d_{KS}(p_{3},\mathbf{x}_{3}) + 2d_{KS}(p_{2},\mathbf{x}_{2}) >\hat{d}\bigg)\\
	& \leq {P}\bigg(d_{KS}(p_{1},\mathbf{x}_{1})> \frac{\hat{d}}{4}\bigg) + {P}\bigg(d_{KS}(p_{3},\mathbf{x}_{3})> \frac{\hat{d}}{4}\bigg)\\
	&\qquad + {P}\bigg(d_{KS}(p_{2},\mathbf{x}_{2})> \frac{\hat{d}}{4}\bigg)\\
	& \leq 6\exp{\big(-\frac{n\hat{d}^{2}}{8}\big)}.
	\end{aligned}
	\end{equation*}
	where $\hat{d}=d_{2}-d_{1}$. By the continuity of the exponential function, we have
	\begin{equation}\label{lemma-to-prove-ks}
	\begin{aligned}
	{P}\bigg(d_{KS}(\mathbf{x}_{1},\mathbf{x}_{3}) \geq d_{KS}(\mathbf{x}_{2},\mathbf{x}_{3})\bigg)\leq 6\exp{\big(-\frac{n(D_{O}-D_{I})^{2}}{8}\big)}.\qedhere
	\end{aligned}
	\end{equation}
\end{proof}

Without loss of generality, assume that the probability that
$\mathbf{y}$ is generated from $p_{k,i_{k}}$ is $\frac{1}{M}$ for all $m\in\{1,\ldots,M\}$ and $i_{m}\in\{1,\ldots,M_{m}\}$. By Lemma \ref{lemma:KStest4} and the union bound, the probability of error is bounded by
\begin{equation}\label{Pe-same-length-ks}
\begin{aligned}
P_{e} &\leq  \sum_{m=1}^{M}\sum_{i_m=1}^{M_{m}}\sum_{m^{\prime}\neq m}\sum_{i_{m^\prime}=1}^{M_{m^{\prime}}} P\bigg(d_{KS}(\mathbf{x}_{m,i_m},\mathbf{y})\geq\\
&\quad  d_{KS}(\mathbf{x}_{m^{\prime},i_{m^\prime}},\mathbf{y})\big|\mathbf{y}\sim p_{m,i_{m}},\: i_{m}\in \{1,\ldots,M_{m}\}\big|\bigg)\frac{1}{M}\\
&\leq 6M\exp{\big(-\frac{n(D_{O}-D_{I})^{2}}{8}\big)}.
\end{aligned}
\end{equation}

Thus, the achievable discrimination rate is
\begin{align}
D = \frac{\log e}{8} (D_O-D_I)^2.
\end{align}

\subsection{Proof of Remark \ref{dis_capacity_comp}}\label{proof_capacity_comp}

Here we provide an alternative proof for Remark \ref{dis_capacity_comp}, which is different from that given in [Lemma 2.10 \cite{Tsybakov:2008:INE:1522486}].

By Fano's inequality \cite{Cover06}, we obtain
 \begin{align}
 H(h|\y) \leq 1+ P_e \log (M-1).
 \end{align}
 Since $h$ is uniformly distributed over all the hypotheses, we have that
 \begin{align}
\log (M)&=H(h)=I(h;\y)+H(h|\y) \nonumber \\
 &\leq I(h;\y)+ 1+P_e \log M.\label{eq:fanoBound_comp}
 \end{align}
 Let $P_h(h)$, $P_{\y}(\y)$, and $P_{h,\y}(h,\y)$ represent the marginal and joint distributions of $h$ and $\y$. Recall that we represent the likelihood function of $\y$ under $m$ as $P(\y|h)=p_h(\y)$. The mutual information between $h$ and $\y$ can be expressed in terms of likelihood functions as
 \begin{align}
 I(h;\y)  =&\sum_{h=1}^{M} \sum_{\y} P_{h,\y}(h,\y) \log \frac{P_{h,\y}(h,\y)}{P_h(h)P_{\y}(\y)} \nonumber\\
 =&\frac{1}{M}\sum_{h=1}^{M}\sum_{\y} p_h(\y) \log \frac{p_h(\y)}{P_{\y}(\y)} \nonumber\\
 =&\frac{1}{M}\sum_{h=1}^{M}\sum_{\y} p_h(\y) \log \frac{p_h(\y)}{\sum_{h^\prime=1}^{M}\frac{1}{M}p_{h^\prime}(\y)} \nonumber \\
 =&\frac{1}{M}\sum_{h=1}^{M}\sum_{\y} p_h(\y) \nn\\
 &\cdot\left[ \log p_h(\y)-\log \sum_{h^\prime=1}^{M}\frac{1}{M}p_{h^\prime}(\y) \right]
 \end{align}
 Applying Jensen's inequality, the mutual information can be further upper bounded as
 \begin{align}
 I(h;\y)  \leq & \frac{1}{M}\sum_{h=1}^{M}\sum_{\y} p_h(\y)\nn\\
 & \cdot \left[ \log p_h(\y)- \sum_{h^\prime=1}^{M}\frac{1}{M}\log p_{h^\prime}(\y) \right]
 \end{align}
 Simplifying, we finally have
 \begin{align}
 I(h;\y)   \leq & \frac{1}{M}\sum_{h=1}^{M}\sum_{\y} p_h(\y)\nn\\
 &\cdot \left[  \sum_{h^\prime=1}^{M}\frac{1}{M}\log p_h(\y)- \sum_{h^\prime=1}^{M}\frac{1}{M}\log p_{h^\prime}(\y) \right] \nonumber \\
 =& \frac{1}{M}\frac{1}{M}\sum_{h=1}^{M}\sum_{h^\prime=1}^{M}\sum_{\y}p_h(\y)\log \frac{p_h(\y)}{p_{h^\prime}(\y)} \nonumber \\
 =& \frac{1}{M}\frac{1}{M}\sum_{h=1}^{M}\sum_{h^\prime=1}^{M}n D_{KL}(p_h\|p_{h^\prime})\nonumber \\
 =& n \mE_{h,h^\prime} D_{KL}(p_h\|p_{h^\prime}).\label{eq:infoBound_comp}
 \end{align}
 where $h^\prime$ has the same distribution as $h$, but is independent from $h$. Substituting (\ref{eq:infoBound_comp}) into the (\ref{eq:fanoBound_comp}), we obtain
 \begin{align}
 \log M&\leq n \mE_{h,h^\prime} D_{KL}(p_h\|p_{h^\prime})+ 1+\log M P_e
 \end{align}
 which implies that
 \begin{align}
 \frac{\log M}{n}&\leq \frac{ \mE_{h,h^\prime} D_{KL}(p_h\|p_{h^\prime})}{1-P_e} + \frac{1}{n(1-P_e)}.
 \end{align}
 Since $M = 2^{nD}$ and the above right-hand side is independent on $D$, we can have
 \begin{align}
 D&\leq \frac{ \mE_{m,m^\prime} D_{KL}(p_m\|p_{m^\prime})}{1-P_e} + \frac{1}{n(1-P_e)}.
 \end{align}
 Thus, for any test that satisfies $P_e \rightarrow 0$ as $n \rightarrow \infty$, $D\leq {\limsup_{M\rightarrow \infty}}E_{m,m^\prime} D_{KL}(p_m\|p_{m^\prime})$ as $n \rightarrow \infty$. Therefore, the discrimination capacity $\bar{D}$ is upper bounded by
 \begin{align}
 \bar{D} \leq {\limsup_{M\rightarrow \infty}}\mE_{m,m^\prime} D_{KL}(p_m\|p_{m^\prime}).
 \end{align}

{\subsection{Sketch of the Proof for \eqref{Pe-multi-length-MMD} and \eqref{Pe-multi-length-KS}}	\label{multi-length-proof-sketch}
To prove \eqref{Pe-multi-length-MMD}, we follow the steps to obtain \eqref{Pe-same-length-mmd}. Note that now $\x^\alpha:=\{\x_{m,i_m},\x_{m^\prime,i_{m^\prime}},\y\}$ consists of {$n+\gamma_{m}(n)+\gamma_{m^\prime}(n)$} data samples, and
\begin{align}
|\bigtriangleup_m(\x^\alpha_{-s}, \x^\alpha_s)- \bigtriangleup_m(\x^\alpha_{-s}, \tilde{\x}^\alpha_s)| \le \frac{8\K}{n^\prime},
\end{align}
where $n^\prime \in \{n, \gamma_{m}(n), \gamma_{m^\prime}(n)\}$ and the correponding choice is based on the location of $\x_s^\alpha$. Then, we can write
\begin{align}
& P\left(\text{MMD}^2(\x_{m,i_m},\y) \ge \text{MMD}^2(\x_{m^\prime,i_{m^\prime}},\y)\right) \nn\\
&=  P(\text{MMD}^2(\x_{m,i_m},\y)- \text{MMD}^2(\x_{m^\prime,i_{m^\prime}},\y) \geq 0)\nn \\
&=	P\left(\bigtriangleup_m(\x^\alpha)-\mathbb{E}[\bigtriangleup_m(\x^\alpha)]\ge -\mathbb{E}[\bigtriangleup_m(\x^\alpha)] \right) \nn \\
&\le	 P\left(\bigtriangleup_m(\x^\alpha)-\mathbb{E}[\bigtriangleup_m(\x^\alpha)]\ge D_O-D_I \right) \nn \\
& \le \exp \left(-\frac{2(D_O-D_I)^2}{64\K^2\left(\frac{1}{n}+\frac{1}{\gamma_{m}(n)}+\frac{1}{\gamma_{m^\prime}(n)}\right)}\right)\nn\\
&\le \exp \left(-\frac{\min\{n,\gamma_{\min}(n)\}(D_O-D_I)^2}{{96\K^2}}\right).
\end{align}
Thus, it yields
\begin{align}
P_e & =P\left(\exists m^\prime \ne m, i_m^\prime \in I_1^{M_m^\prime},\bigtriangleup_m(\x^\alpha)\ge 0, \forall i_m \in I_1^{M_m}\right)\nn\\
& \le \frac{1}{M} \sum_{m=1}^{M}\sum_{m^{\prime}\neq m} \min_{i_m \in I_1^{M_m}} \exp \left(-\frac{n(D_O-D_I)^2}{{96\K^2}}\right)\nn \\
&\leq M\exp \left(-\frac{\min\{n,\gamma_{\min}(n)\}(D_O-D_I)^2}{{96\K^2}}\right).
\end{align}

To prove \eqref{Pe-multi-length-KS}, we follow the steps to obtain \eqref{Pe-same-length-ks}. Note that if the sequences $\mathbf{y} ,\mathbf{x}_{m,i_m},\mathbf{x}_{m^{\prime},i_{m^\prime}}$ have length of $n, \gamma_{m}(n), \gamma_{m^\prime}(n)$ respectively,
we can obtain
\begin{align}
P&\bigg(d_{KS}(\mathbf{x}_{m,i_m},\mathbf{y})\geq d_{KS}(\mathbf{x}_{m^{\prime},i_{m^\prime}},\mathbf{y})\bigg)\nonumber\\
\le& 2\exp{\big(-\frac{n(D_{O}-D_{I})^{2}}{8}\big)}+2\exp{\big(-\frac{\gamma_{m}(n)(D_{O}-D_{I})^{2}}{8}\big)}\nonumber \\&+2\exp{\big(-\frac{\gamma_{m^\prime}(n)(D_{O}-D_{I})^{2}}{8}\big)}\nn\\
\le & 6\exp{\big(-\frac{\min\{n, \gamma_{\min}(n)\}(D_{O}-D_{I})^{2}}{8}\big)}.
\end{align}
Thus, it yields
\begin{equation}
\begin{aligned}
P_{e} &\leq  \sum_{m=1}^{M}\sum_{i_m=1}^{M_{m}}\sum_{m^{\prime}\neq m}\sum_{i_{m^\prime}=1}^{M_{m^{\prime}}} P\bigg(d_{KS}(\mathbf{x}_{m,i_m},\mathbf{y})\geq\\
&\quad  d_{KS}(\mathbf{x}_{m^{\prime},i_{m^\prime}},\mathbf{y})\big|\mathbf{y}\sim p_{m,i_{m}},\: i_{m}\in \{1,\ldots,M_{m}\}\big|\bigg)\frac{1}{M}\\
&\leq 6M\exp{\big(-\frac{\min\{n, \gamma_{\min}(n)\}(D_{O}-D_{I})^{2}}{8}\big)}.
\end{aligned}
\end{equation}
}